\documentclass[runningheads]{llncs}
\usepackage[T1]{fontenc}
\usepackage{amsfonts}
\usepackage{amssymb}
\usepackage{comment}
\usepackage[inline]{enumitem}
\usepackage{multicol}
\usepackage{stmaryrd}
\usepackage{bm}
\usepackage{mathrsfs}
\usepackage{textcomp} 
\usepackage{mathtools}
\usepackage{graphicx}
\makeatletter  
\def\th@plain{%
  \theorem@notefont{}
  \itshape 
}
\def\th@definition{%
  \theorem@notefont{}
  \normalfont 
}
\makeatother

\newcommand{\f}{\varphi}

\newcommand{\p}{\alpha}
\newcommand{\pp}{\beta}

\newcommand{\leftb}{\scalebox{0.9}{\normalfont\texttt{[}}}
\newcommand{\rightb}{\scalebox{0.9}{\normalfont\texttt{]}}}

\newcommand{\B}[1]{\leftb #1 \rightb}

\newcommand{\dia}[1]{\langle#1\rangle}

\begin{document}
\title{On Graded Concurrent PDL}
%
%
\author{CHUN-YU LIN\inst{1}}
\authorrunning{C.Y.LIN}
%
\institute{Department of Logic, Faculty of Arts, Charles University \\Institute of Computer Science, Pod Vod\'arenskou v$\check{e}\check{z}$\'i 271/2
\email{lin@cs.cas.cz}\\
\url{https://www.cs.cas.cz/staff/lin/en} }
\maketitle              
\begin{abstract}
Propositional Dynamic Logic, $\mathsf{PDL}$, is a modal logic designed to formalize the reasoning about programs. By extending accessibility between states to states and state sets, concurrent propositional dynamic logic $\mathsf{CPDL}$, is introduced to include concurrent programs due to Peleg and Goldblatt. We study a many-valued generalization of $\mathsf{CPDL}$ where the satisfiability and the reachability relation between states and state sets are graded over a finite \L ukasiewicz chain. Finitely-valued dynamic logic has been shown to be useful in formalizing reasoning about program behaviors under uncertainty. We obtain completeness results for all finitely valued $\mathsf{PDL}$. 
\end{abstract}
\section{Introduction}
Propositional dynamic logic, $\mathsf{PDL}$, has been introduced as a logic of imperative computer programs \cite{FischerLadner1979}, but it has found many applications in formalizing reasoning about actions in general. Many-valued versions of $\mathsf{PDL}$ aim at formalizing reasoning about action in contexts where imprecise concepts are involved (e.g.,~\cite{Sedlar2021c}). For example, goals may be specified using vague notions, or weights may be attached to actions depending on the amount of resources their execution consumes. The concurrent propositional dynamic logic $\mathsf{CPDL}$ \cite{Goldblatt1992a,Peleg1987} extends $\mathsf{PDL}$ with an operator representing the parallel execution of two actions: $\B{\p \cap \pp}\f$ means that a successful parallel execution of action $\p$ and $\pp$ is guaranteed to make $\f$ true. 

In this paper, we outline a version with many values of $\mathsf{CPDL}$. Our logic is based on propositional dynamic models in which formulas and accessibility relations are evaluated in finite {\L}ukasiewicz chains; therefore, we extend the framework of \cite{Sedlar2021c} with the parallel execution operator. Our main technical result is a soundness and completeness theorem for logic. 
\section{Concurrent PDL over \bf\L$_n$}

Let $\Pi$ be a countable set of program variables $\pi_0,\pi_1.\ldots$ which are used to denote the atomic programs. The intrinsic meaning of atomic programs is not examined further. Instead, we concentrate on the complex programs generated by operations on the given ones. Let $\textbf{P}$ be a countable set of propositional variables. Let \L$_n$ be a finite \L ukasiewicz chain. The language $\mathcal{L}_n$ for concurrent propositional dynamic logic can be separated into two categories-the set of programs $\Pi_n$ and formulas $\Phi_n$ defined mutually recursively in Barckus-Naur form as follows:
\begin{itemize}
    \item $\Pi_n$\mbox{   } $\pi \coloneqq \bar{\pi} \;|\; \pi_0 \cup \pi_1 \;|\; \pi_0 \cap \pi_1 \;|\; \pi_0 ; \pi_1 \;| \; \pi^{*}\;|\; \varphi ? $
    \item $\Phi_n$\mbox{   } $\varphi \coloneqq p\;|\; \bar{c}\;|\; \varphi_0 \vee \varphi_1 \;|\; \varphi_0 \wedge \varphi_1 \;|\; \varphi_0 \to \varphi_1 \;| \; [\pi] \varphi\;|\; \langle \pi \rangle \varphi$
\end{itemize}
where $\bar{\pi} \in \Pi$, $p \in \textbf{P}$ and $c \in$ \L$_n$. From \cite{goldblatt1992parallel}, we know that two modal operators $[\pi]$ and $\langle \pi \rangle$ are not interdefinable via $\neg$ in concurrent PDL.

The intended meanings for the operations on the programs are as follows. 
\begin{itemize}
    \item $\pi_0 ; \pi_1$\mbox{   } execute $\pi_0$ then $\pi_1$,
    \item $\pi_0 \cup \pi_1$\mbox{   } execute $\pi_0$ or $\pi_1$ non-deterministically,
    \item $\pi_0 \cap \pi_1$\mbox{   } execute $\pi_0$ and $\pi_1$ concurrently,
    \item $\pi^{*}$ \mbox{   } execute $\pi$ for some finite number of times,
    \item $\varphi ?$\mbox{   } test $\varphi$ : if $\varphi$ is true, then continue; otherwise fail.
\end{itemize}
In the context of concurrency, the results of execution of an initial state $s$ will be a set of states $T$ rather than a single state. Therefore, the accessibility relation on a set $S$ in Kripke semantics of propositional dynamic logic should be generalized to a set of pair $\langle s,T \rangle$, with $s \in S$ and $T \subseteq S$. We then call the graded accessibility relation as a reachable $\textbf{\L}_n$-relation.
\begin{definition}
A {\it reachable $\textbf{\L}_n$-relation} on a set $S$ is a function from $S \times \mathcal{P}(S)$ to $\textbf{\L}_n$.
\end{definition}
The operations on reachable $\textbf{\L}_n$-relation are defined as another reachable $\textbf{\L}_n$-relations in the following paragraph.
\begin{definition}
Let $R , Q$ be two reachable $\textbf{\L}_n$-relations on a set $S$, $s\in S$ and $T,W \subseteq S$.
\begin{itemize}
    \item $\iota(s,T)=\begin{cases} 1 & \text{if } T =\{ s \}\\
                                    0 & \text{otherwise}  \\           \end{cases}$,
    \item $(R\cup Q)(s,T)= \bigvee_{t \in T}(R(s,T)\vee Q(s,T) )$,
    \item $(R \circ Q)(s,T)=\bigvee_{U \subseteq S} \bigvee\{  R(s,U)\odot \bigodot_{u\in U}Q(u,T_u)| T = \bigcup_{u \in U}T_u  \}$,
    \item $R^{(0)}=\iota$, $R^{(n+1)}=\iota \cup (R\circ R^{(n)})$, and $R^{*}=\bigvee\{ R^{(n)}| n \geq 0 \}$,
    \item $(R \otimes Q)(s,T\cup W)=R(s,T)\odot Q(s,W)$.
\end{itemize}
\begin{lemma}
For any reachable $\textbf{\L}_n$-relations $Q,Q',R,R'$ on a set $S$, we have the following properties.
\begin{enumerate}
    \item $Q \leq Q'$ implies $R \circ Q \leq R\circ Q'$
    \item $(R \cup R') \circ Q = R \circ Q \cup R' \circ Q $
    \item $R^{(n)} \leq R^{(n+1)}$
\end{enumerate}
\end{lemma}
\end{definition}
A {\it $\textbf{\L}_n$-valuation} $\nu$ is any non-modal homomorphism from $\mathcal{L}_n$ to $\textbf{\L}_n$, i.e. $\nu : \mathcal{L}_n \to \textbf{\L}_n$ such that $\nu(\bar{c})= c,$ and $\nu(\varphi \star \psi) = \nu(\varphi) \ast \nu(\psi)$ where $\star , \ast \in \{\wedge,\vee,\to \}$ denote the Boolean operations and $\textbf{L}_n$-operations respectively. 
\begin{definition}
Let $\Theta$ be a set of formulas and $\varphi$ be a formula in $\mathcal{L}_n$. We say that $\varphi$ is a $\textbf{\L}_n$-semantics consequence of $\Theta$ if $$\nu [\Theta]=\{ 1 \} \text{ implies } \nu(\varphi)=1 $$ for any $\textbf{\L}_n$-valuation $\nu$. In this case, we write $\Theta \Vdash_{\textbf{\L}_n} \varphi$.
\end{definition}

We then introduce the world semantics for concurrent PDL. 

\begin{definition}
A {\it $\textbf{\L}_n$-frame} is a pair $\mathcal{F}=\langle S , R \rangle$, where S is a nonempty set and $R : \Pi \to \textbf{\L}_n^{S\times\mathcal{P}(S)}$. A {\it $\textbf{\L}_n$-model} is $M=\langle S,R,V \rangle$ where $\langle S,R \rangle$ is a $\textbf{\L}_n$-frame and $V :\textbf{P} \to {\textbf{\L}_n}^S $. Given a $\textbf{\L}_n$-model $M$, the {\it $M$-interpretation} is a function $\mathcal{I}_M : ((\Pi_n\times S \times \mathcal{P}(S)) \cup (\Phi_n \times S)) \to \textbf{\L}_n$ such that for any $\bar{\pi}\in \Pi, \pi_0,\pi_1 \in \Pi_n,s \in S, T \in \mathcal{P}(S) , p\in \textbf{P}, c \in \textbf{\L}_n$, and $\varphi, \psi \in \mathcal{L}_n$ :
\begin{itemize}
    \item $\mathcal{I}_M(\bar{\pi},s,T)=R(\bar{\pi})(s,T)$,
    \item $\mathcal{I}_M(p,s)=V(p)(s)$,
    \item $\mathcal{I}_M(\bar{c},s)= c $,
    \item $\mathcal{I}_M(\varphi \star \phi, s)=\mathcal{I}_M(s) \ast \mathcal{I}_M(s) $ where $\star , \ast \in \{\wedge,\vee,\to \}$ denote the boolean operations and $\textbf{L}_n$-operations respectively,
    \item $\mathcal{I}_M([\pi]\varphi,s)= \bigwedge_{T \subseteq S}(\mathcal{I}_M(\pi,s,T) \to \bigwedge_{t \in T} \mathcal{I}_M(\varphi,t))$,
    \item $\mathcal{I}_M(\langle \pi \rangle \varphi,s)=\bigvee_{T \subseteq S}(\mathcal{I}_M(\pi,s,T) \odot \bigwedge_{t \in T} \mathcal{I}_M(\varphi,t))$\\
    Denote $\mathcal{I}_M(\pi,s,T)$ as a $\textbf{L}_n$-reachable relation $R_{\pi}(s,T)$
    \item $\mathcal{I}_M(\pi_0 \cup \pi_1,s,T)= ( R_{\pi_0} \cup R_{\pi_1})(s,T)$,
    \item $\mathcal{I}_M(\pi_0;\pi_1,s,T)=(R_{\pi_0}\circ R_{\pi_1})(s,T)$,
    \item $\mathcal{I}_M(\pi_0 \cap \pi_1,s,T)= (R_{\pi_0}\otimes R_{\pi_1})(s,T)$,
    \item $\mathcal{I}_M(\pi^{*},s,T)= (R_{\pi})^{*}(s,T)$,
    \item $\mathcal{I}_M(\varphi ?,s,T)= \begin{cases}
    \mathcal{I}_M(\varphi,s) & \text{if } T=\{ s\}\\
    0 & \text{otherwise.} \\
    \end{cases}$
\end{itemize}
\end{definition}
Clearly, fixed an $s \in S$, a $M$-interpretation is a $\textbf{\L}_n$-valuation. A formula $\varphi$ is called valid in a $\textbf{\L}_n$-model $M$ if $\mathcal{I}_M(\varphi,s)=1$ for all $s\in S$. 

\section{Finite MV-chains}
We briefly introduce the definitions and basic properties of finite MV-chains in this section. From more detailed introduction, we refer to the handbooks\cite{cintula2011handbook}.
\begin{definition}
A finite MV-chain is an algebraic structure $$\text{\L}_n = \langle A,\lor,\land,\to,\odot,\bar{0},\bar{1} \rangle$$ with $|A|=n$ such that : 
\begin{itemize}
    \item $\langle A,\lor,\land,\bar{0},\bar{1} \rangle$ is a finite bounded lattice,
    \item $\langle A, \odot, \bar{1} \rangle$ is a finite commutative monoid,
    \item Define the total ordering $\leq$ as $a\leq b$ iff $a \land b=b$ iff $a \lor b =a$,
    \item $\odot$ is residuated with $\to$, i.e. for all $a,b,c \in A$, $a\odot b \leq c$ iff $b \leq a \to c$,
    \item $(a \to b) \vee (b \to a) = \bar{1}$ for all $a ,b \in A$.
    \item $a \wedge b = a\odot(a \to b)$ for all $a ,b \in A$. \item Define $\neg a \coloneqq a \to \bar{0}$ satisfying $\neg \neg a = a$ for all $a \in A$.
\end{itemize}
\end{definition}
The above definition is an abstract formulation. According to Corollary 3.5.4 in \cite{cignoli2013algebraic}, any finite MV-chain $\text{\L}_n$ is isomorphic to the following more concrete finite MV-chains $\textbf{\L}_n$ for any $n \geq 2$.
$$\textbf{\L}_n = \langle \{\frac{0}{n-1},\frac{1}{n-1},\ldots,\frac{n-1}{n-1}  \} ,\to,\odot,\sim,\bar{0},\bar{1} \rangle $$
where $a \odot b \coloneqq max\{0,a+b-1 \} $, $a \to b \coloneqq min \{1,1-a+b  \}$, and $\sim a \coloneqq 1-a$ for any $a,b \in \{\frac{0}{n-1},\frac{1}{n-1},\ldots,\frac{n-1}{n-1}  \}$.
Note that $a \vee b \coloneqq (a \to b) \to b (=max\{a,b\})$ and $a \wedge b \coloneqq a \odot(a \to b)(= min \{a,b\})$. The truncated addition $\odot$ is the \L ukasiewicz t-norm. In this article, we will use $\textbf{\L}_n$ as our definition of finite MV-chain.

\section{Proof System}

For each $n \geq 2$, the axiomatic system $\textbf{P\L}_n$ is a Hilbert-style proof system consists the following axiom schemata and rule : 
\begin{itemize}
    \item $ \varphi \to (\psi \to \varphi)$,
    \item $ (\varphi \to \psi) \to ((\psi \to \chi) \to (\varphi \to \chi)) $,
    \item $( (\varphi \to \psi) \to \psi) \to ((\psi \to \varphi) \to \varphi) $,
    \item $(\neg \psi \to \neg \varphi) \to (\varphi \to  \psi)$,
    \item $\overline{c \ast d} \leftrightarrow \bar{c} \star \bar{d} $ where $\star , \ast \in \{\wedge,\vee,\to \}$ denote the boolean operations and $\textbf{L}_n$-operations respectively,
    \item (MP): from $\varphi $ and $\varphi\to \psi$ infer $\psi$.
\end{itemize}
We give the definition of a formula $\varphi$ being derivable in $\textbf{P\L}_n$ from a set of formulas $\Theta$.
\begin{definition}
A formula $\varphi $ in $\mathcal{L}_n$ is derivable from a set of formulas $\Theta$ in $\textbf{P\L}_n$ if there exists a finite sequence of formulas $\varphi_0,\varphi_1,\ldots,\varphi_m$ such that $\varphi_m$ is $\varphi$ and each $\varphi_i$ for $i<m$ is either an instance of an axiom schemata, a member of $\Theta$, or follows from $\varphi_j$ and $\varphi_k$ using MP for $j,k<i$.
\end{definition}
We use $\Theta \vdash_{\textbf{P\L}_n} \varphi$ to denote that $\varphi$ is derivable in $\textbf{P\L}_n$ from a set of formulas $\Theta$.

A slight modification of the proof of Proposition 6.4.5 in \cite{cintula2011handbook} gives the following theorem.
\begin{theorem}\label{theorem1}
Let $\Theta$ be a set of formulas and $\varphi$ be a formula in $\mathcal{L}_n$. We have $$\Theta \Vdash_{\textbf{\L}_n} \varphi \text{ iff } \Theta \vdash_{\textbf{P\L}_n} \varphi
$$ 
\end{theorem}
Extend $\textbf{P\L}_n$ with axiom schemata and rules for modal formulas, we get another Hilbert-style proof system called $\textbf{D\L}_n$.
\begin{definition}
For each $n \geq 2$, define $\textbf{D\L}_n$ to be the Hilbert-style proof system that extends $\textbf{P\L}_n$ with the following axiomatic schemata and rules :
\begin{itemize}
    \item $[\pi]\bar{1}$,
    \item $[\pi]\varphi \wedge [\pi]\psi \to [\pi](\varphi \wedge \psi)$,
    \item $[\pi](\bar{c}\to \varphi ) \leftrightarrow (\bar{c}\to [\pi] \varphi)$,
    \item $[\pi](\varphi \to \bar{c})\leftrightarrow (\langle \pi \rangle \varphi \to \bar{c})$,
    \item $[\pi_0 ;\pi_1 ]\varphi \leftrightarrow [\pi_0][\pi_1]\varphi$,
    \item $[\pi_0 \cup \pi_1] \varphi \leftrightarrow [\pi_0]\varphi \wedge [\pi_1]\varphi$,
    \item $[\pi_0 \cap \pi_1]\varphi \leftrightarrow (\langle \pi_0\rangle \bar{1} \to [\pi_1]\varphi) \wedge(\langle \pi_1\rangle  \bar{1} \to [\pi_1]\varphi )$,
    \item $[\pi^{*}]\varphi \rightarrow \varphi \wedge [\pi][\pi^{*}]\varphi$,
    \item $[\pi^{*}](\varphi \to [\pi]\varphi)\to (\varphi \to [\pi^{*}]\varphi)$
    \item $[\varphi?]\psi \leftrightarrow(\varphi \to \psi)$,
    \item $\langle \pi_0;\pi_1\rangle\varphi \leftrightarrow \langle \pi_0\rangle\langle \pi_1\rangle\varphi$,
    \item $\langle \pi_0 \cup \pi_1\rangle \varphi \leftrightarrow \langle \pi_0\rangle \varphi \vee \langle \pi_1\rangle \varphi$,
    \item $\langle \pi_0 \cap \pi_1\rangle \varphi \leftrightarrow \langle \pi_0\rangle \varphi \wedge \langle \pi_1\rangle \varphi$,
    \item $\varphi \vee \langle \pi\rangle\langle \pi^{*}\rangle \varphi \rightarrow\langle \pi^{*}\rangle \varphi$,
    \item $[\pi^*](\langle \pi \rangle\varphi\to \varphi)\to (\langle\pi^*\rangle\varphi\to \varphi)$,
    \item $\langle \varphi ?\rangle \psi \leftrightarrow (\varphi \wedge \psi)$,
    \item $[\pi]\bar{0} \vee \langle \pi\rangle \bar{1} $.
\end{itemize}
\end{definition}
The first four axioms and monotonicity rule are adopted from \cite{vidal2017finite} which characterizes the minimum many-valued bimodal logics over finite residuated lattices. The rest is the standard axiomatization of concurrent PDL \cite{goldblatt1992parallel}. We let $Thm_n$ to be the set of theorems of $\textbf{D\L}_n$, i.e. 
$$Thm_n= \{ \varphi \in \mathcal{L}_n\;|\;\vdash_{\textbf{D\L}_n} \varphi \}$$
\section{Canonical Model}
In this section, we introduce the canonical models and filtration technique for connecting finite models and typical models.
For $n \in \mathbb{N}$, $\mathcal{F}_n =\langle S^n,\mathcal{I}^n \rangle$ consists of the set $$S^n= \{ s : \mathcal{L}_n \to \textbf{\L}_n \;|\; s[Thm_n]=\{\bar{1}\}\}$$ and $$\mathcal{I}_n:\Pi_n\times S^n\times P(S^n) \cup \Phi_n\times S^n \to \textbf{\L}_n $$ such that $\mathcal{I}^n(\varphi,s)=s(\varphi)$ for all $s \in S^n$, $$
\mathcal{I}^n(\pi,s,T)= \bigwedge_{\varphi \in \Phi_n}\{(s([\pi]\varphi) \to\wedge_{t\in T}t(\varphi))\wedge(\wedge_{t\in T}t(\varphi) \to s(\dia{\pi} \varphi)) \}.$$ The {\it canonical model} in  is $\mathcal{M}^n_c= \langle S^n,E^n,V^n \rangle$ where $E^n(\pi)(s,T)=\mathcal{I}^n(\pi,s,T)$ and $V^n(p)(s)=s(p)$
\begin{definition}
A set of formulas $\Gamma$ of $\mathcal{L}_n$ is called {\it Fisher-Ladner clsoed} if 
\begin{itemize}
    \item $\Gamma$ is closed under subformulas,
    \item $[\pi_0 \cup\pi_1]\varphi \in \Gamma$ implies $[\pi_0]\varphi,[\pi_1]\varphi \in \Gamma$,
    \item $[\pi_0\cap \pi_1]\varphi\in \Gamma$ implies $[\pi_0]\varphi,[\pi_1]\varphi,[\pi_0]\bar{1},[\pi_1]\bar{1} \in \Gamma$,
    \item $[\pi_0;\pi_1]\varphi \in \Gamma$ implies $[\pi_0][\pi_1]\varphi\in \Gamma$,
    \item $[\pi^{*}]\varphi \in \Gamma$ implies $[\pi][\pi^*]\varphi \in \Gamma$,
    \item $[\psi?]\varphi \in \Gamma$ implies $\psi \to \varphi \in \Gamma$,
    \item$\langle\pi_0 \cup\pi_1\rangle\varphi \in \Gamma$ implies $\langle \pi_0\rangle\varphi,\langle \pi_1 \rangle \varphi \in \Gamma$,
    \item  $\langle \pi_0\cap \pi_1\rangle \varphi\in \Gamma$ implies $\langle\pi_0\rangle\varphi,\langle\pi_1\rangle\varphi \in \Gamma$,
    \item $\langle \pi_0;\pi_1 \rangle \varphi \in \Gamma$ implies $\langle \pi_0\rangle\langle \pi_1\rangle \varphi\in \Gamma$,
    \item $\langle\pi^{*}\rangle\varphi \in \Gamma$ implies $\langle \pi\rangle\langle\pi^*\rangle\varphi \in \Gamma$,
    \item $\dia{\psi?}\varphi \in \Gamma$ implies $\psi \wedge \varphi \in \Gamma$
\end{itemize}
\end{definition}
The closure of a set $\Gamma$ of formulas is the smallest closed set containing $\Gamma$ as a subset. Also, we write $FL(\varphi)$ as the closure of the set $\{\varphi \}$. Given any $S^n$ and $s,t \in S^n$, we define a relation $\sim_{\Gamma}$ with respect to a $\Gamma$ as follows : $$s\sim_{\Gamma}t \iff \forall \varphi \in \Gamma (s(\varphi)=t (\varphi)) .$$ The equivalence class of $s$ under $\sim_{\Gamma}$ is defined as $|s|_{\Gamma}= \{t \in S^n\;|\; s\sim_{\Gamma} t \}$.
\begin{definition}
Let $\Gamma$ be a finite closed subset of $\mathcal{L}_n$. The filtration of $\mathcal{M}_c^n$ through $\Gamma$ is $\mathcal{M}_{c,\Gamma}^n= \langle S^n_{\Gamma},E^n_{\Gamma},V^n_{\Gamma}\rangle$ where $S^n_{\Gamma}=\{|s|_{\Gamma}\;|\;s \in S^n \}$,$|T|_{\Gamma}=\{|t|_{\Gamma}\;|\;t \in T \}$ for $T \subseteq S^n$, $E^n_{\Gamma}(\pi)(|s|_{\Gamma},|T|_{\Gamma})$  is $$\bigwedge_{\varphi \in \Phi_n}\{(s([\pi]\varphi) \to\wedge_{t\in T}t(\varphi))\wedge(\wedge_{t\in T}t(\varphi) \to s(\dia{\pi}\varphi))\;|\;[\pi]\varphi,\dia{\pi}\varphi\in \Gamma \}$$ $V^n_{\Gamma}(p)(|s_{\Gamma}|)=\begin{cases}
    V^n(p)(s) & \text{if } p\in\Gamma\\
    0 & \text{otherwise.} \\
    \end{cases}$
\end{definition}
\section{Soundness and Completeness}
In this section, we demonstrate that for each $n >1 $, the proof system $\textbf{D\L}_n$ is sound and complete with respect to the filtraion models.
To prove completeness, we need to show that for all $\varphi$ and all $s\in S^n$, $\mathcal{I}^n(\varphi,s)=\mathcal{I}^n_{\Gamma}(\varphi,|s|_{\Gamma})$. To achieve this goal, we define $\Pi_{\Gamma,n}$ as the smallest set of program commands that includes all atomic programs, test occurring in members of $\Gamma$, and is closed under program operations $;,\cup,\cap,$ and $*$. Then we define $\mathcal{I}^n_{\Gamma}: \Pi_{\Gamma,n}\times S^n_{\Gamma}\times P(S^n)\cup \Gamma\times S^n_{\Gamma}\to \text{\L}_n$ as for $\textbf{\L}_n$-models. 
\begin{lemma}
For all $n \in \omega$, $\pi \in \Pi_n$, and all $s \in S_n$, we have the following identities:
\begin{enumerate}
    \item For all $\varphi \in \Phi_n$, $s([\pi]\varphi)=\bigwedge_{T\subseteq S^n}\{\mathcal{I}^n(\pi,s,T) \to \wedge_{t \in T}t(\varphi)\}$, $s(\dia{\pi}\varphi)=\bigvee_{T \subseteq S^n}\{ \mathcal{I}^n(\pi,s,T)\odot\ \wedge_{t \in T}t(\varphi) \}$,
    \item For all $T \subseteq S^n$, $\mathcal{I}^n(\pi,s,T)= \bigwedge_{\psi \in \Phi_n}\{\wedge_{t\in T}t(\psi)\;|\;s([\pi]\psi)=1 \mbox{ and }s(\dia{\pi}\psi)=1 \} $.
\end{enumerate}
\end{lemma}

\begin{lemma}\label{lemma3}
$s(\varphi)=1$ for all $s \in S^n$ iff $\varphi \in theorem_n$.
\end{lemma}
\begin{proof}
The if direction holds by the definition of $S^n$. Suppose that $\nvdash_{\textbf{D\L}_n}\varphi$, then $theorem_n\nvdash_{\textbf{D\L}_n} \varphi$. From Theorem \ref{theorem1}, there exists a $\textbf{\L}_n$-homomorphism $h$ such that $h[theorem_n]\subseteq \{1\}$ but $h(\varphi)<1$ which is a contradiction.
\end{proof}
\begin{lemma}\label{lemma4}
For all $n\in \omega$ and all finite closed $\Gamma\subseteq \Phi_n$ 
\begin{enumerate}
    \item $\mathcal{I}^n(\pi,s,T) \leq \mathcal{I}^n_{\Gamma}(\pi,|s|_{\Gamma},|T|_{\Gamma})$,
    \item If $[\pi]\varphi\in\Gamma$, then $s([\pi]\varphi)\leq \mathcal{I}^n_{\Gamma}(\pi,|s|_{\Gamma},|T|_{\Gamma})\to \wedge_{t \in T}t(\varphi)$,
    \item If $\dia{\pi}\varphi \in \Gamma$, then $s(\dia{\pi}\varphi)\leq \mathcal{I}^n_{\Gamma}(\pi,|s|_{\Gamma},|T|_{\Gamma})\odot \wedge_{t \in T}t(\varphi)$
\end{enumerate}
\end{lemma}
\begin{theorem}\label{theorem2}
For all $n\in \omega$ and all finite closed $\Gamma\subseteq \Phi_n$, if $\varphi \in \Gamma$, then $\mathcal{I}^n(\varphi,s)=\mathcal{I}^n_{\Gamma}(\varphi,|s|_{\Gamma})$.
\end{theorem}
\begin{proof}
By induction on the complexity of formulas $\varphi$ and using Lemma\ref{lemma4}.
\end{proof}
Now, we are ready to the soundness and completeness theorem.
\begin{theorem}
For all $n >1$ and $\varphi\in \mathcal{L}_n$, $$\Vdash_{\textbf{\L}_n}\varphi \mbox{ iff } \vdash_{\textbf{D\L}_n} \varphi$$
\end{theorem}
\begin{proof}
For the only if direction, we assume that $\nvdash_{\textbf{D\L}_n} \varphi$. Consider the FL-closure of $\{ \varphi\}$ and denote it as $\Gamma$. From lemma \ref{lemma3}, there exists $s \in S^n$ such that $s(\varphi)$. Consider the canonical model $\mathcal{M}^n_{c,\Gamma}$ constructed as above. Using Theorem \ref{theorem2}, we have $\mathcal{M}^n_{c,\Gamma}\nVdash_{\textbf{\L}_n}\varphi$ which contradicts to the assumption.
\end{proof}
\section{Conclusion}
We studied many-valued concurrent propositional dynamic logics through relational models where both statisfaction of formulas and accessibility relations are evaluated in finite MV-chains. We provides a sound and weakly complete axiomatization based on extending the framework from many-valued bimodal logics in \cite{vidal2017finite} and classical concurrent PDL in \cite{goldblatt1992parallel}. We believe this research direction lays the groundwork for future investigations.

For the future research directions, let us mention two here. Firstly, the revision and extension of PDL toward modeling concurrency have been studied in various models such as $\pi$-calculus \cite{benevides2010propositional}, Petri nets \cite{lopes2014propositional}, and operational semantics \cite{acclavio2024propositional}. It would be interesting to study the PDL in the setting of concurrency with imprecise concepts. Secondly,
in light of \cite{sedlar2024completeness} to use the finitely weighted Kleene algebra with tests as an algebraic semantic for graded PDL, it is interesting to explore the algebraic framework of graded concurrent PDL. A first step of this goal would be expanding the concurrent Kleene algebras with tests proposed in \cite{jipsen2014concurrent}. 

\subsubsection{\ackname} 
This work was supported by the Czech Science Foundation grant 22-16111S for the project {\it GRADLACT: Graded Logics of Action} and Charles University research grant GAUK 101724 for the project {\it Zkoumání základu uvažování v racionálních interakcích za nejistých podmínek}.
\newpage
\bibliographystyle{plain}
\bibliography{bibliography.bib}
\end{document}